\documentclass[11pt,a4paper,twocolumn]{article}
\usepackage[latin1]{inputenc}
\usepackage{amsmath}
\usepackage{amsfonts}
\usepackage{amssymb}
\usepackage{pst-all}

\colorlet{darkred}{red!80!black}
\colorlet{darkgreen}{green!70!black}
\colorlet{darkyellow}{yellow!90!black}
\colorlet{darkorange}{orange!88!black}
\colorlet{darkblue}{blue!75!black}

\newcommand\qed{\hfill\ensuremath{\Box}\smallskip}

\newtheorem{theorem}{Theorem}[section]

\newtheorem{corollary}[theorem]{Corollary}

\newtheorem{lemma}[theorem]{Lemma}

\newenvironment{proof}{
        \noindent {\bf Proof: }}{ }

\newcommand\buchi{B\"uchi}
\newcommand\cobuchi{Co-B\"uchi}

\newcommand\A{\mathcal A}
\newcommand\B{\mathcal B}

\newcommand\G{\mathcal G}
\newcommand\scL{\mathcal L}
\newcommand\scP{\mathcal P}

\renewcommand\uplus{\dot{\cup}}
\renewcommand\uplus{\cup}

\begin{document}

\sloppy

\title{Minimising Good-for-Games automata is NP complete%
\footnote{This work was partly supported by the Engineering and Physical Science Research Council (EPSRC) through the grant EP/P020909/1 `Solving Parity Games in Theory and Practice'.}}

\author{Sven Schewe\\University of Liverpool\\\tt sven.schewe@liverpool.ac.uk}

\date{}

\maketitle
\begin{abstract}
This paper discusses the hardness of finding minimal good-for-games (GFG) \buchi, \cobuchi, and parity automata with state based acceptance.
The problem appears to sit between finding small deterministic and finding small nondeterministic automata,
where minimality is NP-complete and PSPACE-complete, respectively.
However, recent work of Radi and Kupferman has shown that minimising \cobuchi\ automata with transition based acceptance is tractable,
which suggests that the complexity of minimising GFG automata might be \emph{cheaper} than minimising deterministic automata.

We show for the standard state based acceptance that the minimality of a GFG automaton is NP-complete for \buchi, \cobuchi, and parity GFG automata.
The proofs are a surprisingly straight forward generalisation of the proofs from deterministic \buchi\ automata:
they use a similar reductions, and the same hard class of languages.
\end{abstract}

\section{Introduction}

Good-for-games (GFG) automata form a useful class of automata that can be used to replace deterministic automata to recognise languages in several settings,
like synthesis \cite{DBLP:conf/csl/HenzingerP06}.
As good-for-games automata sit between deterministic and general nondeterministic automata, it stands to be expected that
the complexity of their minimality also sits between the minimality of deterministic automata (NP-complete \cite{Schewe/10/minimise}) and
nondeterministic automata (which is PSPACE-complete like for nondeterministic finite automata \cite{DBLP:journals/siamcomp/JiangR93}).
It thus came as a surprise when Radi and Kupferman showed that minimising \cobuchi\ automata with transition based acceptance is tractable \cite{DBLP:conf/icalp/RadiK19}.

This raises the question if the difference is in good-for-games automata being inherently simpler to minimise, or if it is a property of
choosing the less common transition based acceptance.
We show that the answer for the more common state based acceptance is that minimising GFG automata is as hard as minimising deterministic automata.

The proof generalises the NP-completeness proof from \cite{Schewe/10/minimise}---and does not extend to the hardness of  minimising transition based automata.

\section{Automata}

\subsection{Nondeterministic Parity Automata}
Parity automata are word automata that recognise $\omega$-regular languages over finite set of symbols.
A \emph{nondeterministic parity automaton} (NPA) is a tuple $\scP = (\Sigma,Q,q_0,\delta,\pi)$, where
\begin{itemize}
\item $\Sigma$ denotes a finite set of symbols,
\item $Q$ denotes a finite set of states,
\item $q_0 \in Q_+$ with $Q_+ = Q \uplus \{\bot,\top\}$ denotes a designated initial state,
\item $\delta: Q_+ \times \Sigma \rightarrow 2^Q_+$ (with $2^Q_+ = 2^Q \uplus \{\{\bot\},\{\top\}\} \setminus \{\emptyset\}$ is a function that maps pairs of states and input letters to either a non-empty set of states, or to $\bot$ (false, immediate rejection, blocking) or $\top$ (true, immediate acceptance)%
\footnote{The question whether or not an automaton can immediately accept or reject is a matter of taste. Often, immediate rejection is covered by allowing $\delta$ to be partial while there is no immediate acceptance.
We allow both---so$\top$ and $\bot$ are not counted as states---but treat them as accepting and rejecting sink states, respectively, for technical convenience.},
such that $\delta(\top,\sigma)=\{\top\}$ and $\delta(\bot,\sigma)=\{\bot\}$ hold for all $\sigma\in \Sigma$, and

\item $\pi: Q_+ \rightarrow P \subset \mathbb{N}$ is a priority function that maps states to natural numbers
(mapping $\bot$ and $\top$ to an odd and even number, respectively), called their \emph{priority}.
% (They are often referred to as colours.)
\end{itemize}

Parity automata read infinite \textbf{input words} $\alpha = a_0 a_1 a_2 \ldots \in \Sigma^\omega$.
(As usual, $\omega = \mathbb N_0$ denotes the non-negative integers.)
Their acceptance mechanism is defined in terms of runs:
A run $\rho =r_0r_1r_2\ldots\in {Q_+}^\omega$ of $\scP$ on $\alpha$ is an $\omega$-word that satisfies $r_0=q_0$ and, for all $i\in \omega$, $r_{i+1}\in\delta(r_i,a_i)$.
A run is called \emph{accepting} if the highest number occurring infinitely often in the infinite sequence $\pi(r_0)\pi(r_1)\pi(r_2)\ldots$ is even, and \emph{rejecting} if it is odd.
An $\omega$-word is \emph{accepted} by $\scP$ if it has an accepting run.
The set of $\omega$-words accepted by $\scP$ is called its \emph{language}, denoted $\scL(\scP)$.
Two automata that recognise the same language are called \emph{language equivalent}.

We assume without loss of generality that $\max \{P\} \leq |Q|+1$.
(If a priority $p \geq 2$ does not exist, we can reduce the priority of all states whose priority is strictly greater than $p$ by $2$ without affecting acceptance.)

\subsection{\buchi\ and \cobuchi\ Automata}
\buchi\ and \cobuchi\ automata---abbreviated NBAs and NCAs---are NPAs where
the image of the priority function $\pi$ is contained in $\{1,2\}$ and $\{2,3\}$, respectively.
In both cases, the automaton is often denoted $\A=(\Sigma,Q,q_0,\delta,F)$,
where $F \subseteq Q_+$ is called (the set of) \emph{final states} and denotes those states with the higher priority $2$.
The remaining states $Q_+\setminus F$ are called \emph{non-final} states.

\subsection{Deterministic and Good-for-Games Automata}
An automaton is called \emph{deterministic} if the domain of the transition function $\delta$ consists only of singletons (i.e.\ is included in $\big\{\{q\}\mid q \in Q_+\big\}$.
For convenience, $\delta$ is therefore often viewed as a function $\bar{\delta}:Q_+ \times \Sigma \rightarrow Q_+$ (with $\delta(q,\sigma)\mapsto \{\bar{\delta}(q,\sigma)\}$).

An nondeterministic automaton is called \emph{good-for-games} (GFG) if it only relies on a limited form of nondeterminism:
GFG automata can make their decision of how to resolve their nondeterministic choices on the \emph{history} at any point of a run---rather
than using the knowledge of the complete word as a nondeterministic automaton normally would---without changing their language.
They can be characterised in many ways, including as automata that simulate deterministic automata.

We use the following formalisation:
a nondeterministic $\scP = (\Sigma,Q,q_0,\delta,\pi)$ is good-for-games if there is function $\nu : q_0 {Q_+}^* \Sigma \rightarrow Q_+$ such that,
for every infinite word $\alpha = a_0 a_1 a_2 \ldots \in \Sigma^\omega$, $\scP$ has an accepting run $\rho'$ if,
and only if, it has an accepting run $\rho =r_0r_1r_2\ldots\in {Q_+}^\omega$ with $r_0 = q_0$ and, for all $i \in \mathbb N_0$,
$r_{i+1} = \nu(r_0,\ldots,r_i;a_i)$.

Broadly speaking, a good-for-games automaton sits in the middle between a nondeterministic and a deterministic automaton:
$\scP$ and $\nu$ together define a deterministic automaton (if such a $\nu$ exists, there is a finite state one),
but as the $\nu$ does not have to be explicit provided, $\scP$ can be more succinct than a deterministic automaton.

% \subsection{Deterministic Finite Automata}
% Finite automata are word automata that recognise the regular languages over finite set of symbols.
% A \emph{deterministic finite automaton} (DFA) is a tuple $\F = (\Sigma,Q,q_0,\bar{\delta},F)$, where $\Sigma$, $Q$, $q_0$, and $\bar{\delta}$ are defined a for DPAs, and $F \subseteq Q \uplus \{\top\}$ is a set of \emph{final states} that contains $\top$ (but not $\bot$).
% 
% Finite automata read finite input words $\alpha = a_0 a_1 a_2 \ldots a_n \in \Sigma^*$.
% Their acceptance mechanism is again defined in terms of runs:
% The unique run $\rho =r_0r_1r_2\ldots r_{n+1} \in {Q_+}^+$ of $\F$ on $\alpha$ is the word that satisfies $r_0=q_0$ and, for all $i\leq n$, $r_{i+1}=\bar{\delta}(r_i,a_i)$.
% A run is called \emph{accepting} if it ends in a final state (and \emph{rejecting} otherwise), a word is \emph{accepted} by $\F$ if its run is accepting, and the set of words accepted by $\F$ is called its \emph{language}, denoted $\scL(\F)$.

\subsection{Automata Transformations \& Conventions}
For an NPA $\B=(\Sigma,Q,q_0,\delta,\pi)$ and a state $q \in Q_+$, we denote with $\B_q=(\Sigma,Q,q,\delta,\pi)$
the automaton resulting from $\B$ by changing the initial state to $q$.
% We also read finite automata at times as \buchi\ (or \cobuchi) automata and \buchi\ (or \cobuchi) automata as finite automata in the constructions, and let DFAs run on infinite words where this is convenient and its meaning is clear in the context.

% Automata define a directed graph whose unravelling from the initial state defines the possible runs.
% For an automaton $\A=(\Sigma,Q,q_0,\bar{\delta},F)$ or $\A=(\Sigma,Q,q_0,\bar{\delta},\pi)$, this is the directed graph $(Q_+,T)$ with $T=\{(p,q) \in Q_+ \times Q_+ \mid \exists \sigma \in \Sigma.\ \bar{\delta}(p,\sigma)=q\}$.
% When referring to the reachable states (which always means reachable from the initial state) and SCCs of an automaton, this refers to this graph.

There are two standard measures for the size of an automaton $\scP = (\Sigma,Q,q_0,\delta,\pi)$: the number $|Q|$ of its states, and the size
$\sum\limits{q\in Q,a \in \sigma}|\delta(q,a)|$ of its transition table.

\section{Main Result}
We show the following theorem.
\begin{theorem}
\label{theo:main}
The following problems are NP-complete.
\begin{enumerate}
 \item Given a good-for-games parity automaton and a bound $k$, is there a language equivalent good-for-games parity automaton with at most $k$ states?
 \item Given a good-for-games \buchi\ automaton and a bound $k$, is there a language equivalent good-for-games parity automaton with at most $k$ states?
 \item Given a good-for-games \cobuchi\ automaton and a bound $k$, is there a language equivalent good-for-games parity automaton with at most $k$ states?
 \item Given a good-for-games \buchi\ automaton and a bound $k$, is there a language equivalent good-for-games \buchi\ automaton with at most $k$ states?
 \item Given a good-for-games \cobuchi\ automaton and a bound $k$, is there a language equivalent good-for-games \cobuchi\ automaton with at most $k$ states?
 \item Given a good-for-games parity automaton and a bound $k$, is there a language equivalent good-for-games parity automaton with at most $k$ entries in its transition table?
 \item Given a good-for-games \buchi\ automaton and a bound $k$, is there a language equivalent good-for-games parity automaton with at most $k$ entries in its transition table?
 \item Given a good-for-games \cobuchi\ automaton and a bound $k$, is there a language equivalent good-for-games parity automaton with at most $k$ entries in its transition table?
 \item Given a good-for-games \buchi\ automaton and a bound $k$, is there a language equivalent good-for-games \buchi\ automaton with at most $k$ entries in its transition table?
 \item Given a good-for-games \cobuchi\ automaton and a bound $k$, is there a language equivalent good-for-games \cobuchi\ automaton with at most $k$ entries in its transition table?
 \item Given a deterministic parity automaton and a bound $k$, is there a language equivalent good-for-games parity automaton with at most $k$ states?
 \item Given a deterministic \buchi\ automaton and a bound $k$, is there a language equivalent good-for-games parity automaton with at most $k$ states?
 \item Given a deterministic \cobuchi\ automaton and a bound $k$, is there a language equivalent good-for-games parity automaton with at most $k$ states?
 \item Given a deterministic \buchi\ automaton and a bound $k$, is there a language equivalent good-for-games \buchi\ automaton with at most $k$ states?
 \item Given a deterministic \cobuchi\ automaton and a bound $k$, is there a language equivalent good-for-games \cobuchi\ automaton with at most $k$ states?
 \item Given a deterministic parity automaton and a bound $k$, is there a language equivalent good-for-games parity automaton with at most $k$ entries in its transition table?
 \item Given a deterministic \buchi\ automaton and a bound $k$, is there a language equivalent good-for-games parity automaton with at most $k$ entries in its transition table?
 \item Given a deterministic \cobuchi\ automaton and a bound $k$, is there a language equivalent good-for-games parity automaton with at most $k$ entries in its transition table?
 \item Given a deterministic \buchi\ automaton and a bound $k$, is there a language equivalent good-for-games \buchi\ automaton with at most $k$ entries in its transition table?
 \item Given a deterministic \cobuchi\ automaton and a bound $k$, is there a language equivalent good-for-games \cobuchi\ automaton with at most $k$ entries in its transition table?
\end{enumerate}
\end{theorem}

The questions are, of course, all very similar. Note, however, that the good-for-games property of the given automaton in Properties 1 through 10 is not checked by the algorithms provided below, and
the complexity of determining GFG-ness is current research (but known to be tractable for \buchi~\cite{DBLP:conf/fsttcs/BagnolK18} and \cobuchi~\cite{DBLP:conf/icalp/KuperbergS15} automata).
The input does not include a witness for its GFG-ness but the result is not reliable if the input NPA is not good-for-games.

The proofs fall into the two categories of inclusion in NP (Section \ref{sec:inNP}), where the small good-for-games automaton can be guessed and
the guess validated with standard simulation games (Corollary \ref{cor:inNP}).

For the NP hardness (Theorem \ref{theo:hard}), it turns out that the known hardness proof for deterministic \buchi\ and \cobuchi\ automata can be adjusted to good-for-games automata, providing hardness for all combinations for our main theorem.

\section{Inclusion in NP}
\label{sec:inNP}

We start with re-visiting a standard simulation game between a \textsf{verifier}, who wants to establish language inclusion through simulation, and a \textsf{spoiler}, who wants to destroy the proof.
Note that the spoiler does not try to disprove language inclusion, but merely wants to show that it cannot be established through simulation.

\subsection{Simulation Game}

For two NPAs $\scP^1 = (\Sigma,Q_1,q_0^1,\delta_1,\pi_1)$ and $\scP^2 = (\Sigma,Q_2,q_0^2,\delta_2,\pi_2)$, we define the `$\scP^2 \;\mathsf{simulates}\; \scP^1$' game, where a \textsf{spoiler} intuitively tries to show that $\scP^1$ accepts a word not in the language of $\scP^2$, as follows.

The game is played on $Q_1 \times Q_2 \cup Q_1 \times \Sigma \times Q_2$ and starts in $(q_0^1,q_0^2)$.
In a state $(q_1,q_2) \in Q_1 \times Q_2$, the \textsf{spoiler} selects a letter $\sigma \in \Sigma$ and a $\sigma$ successor $q_1' \in \delta(q_1,\sigma)$ of $q_1$ for $\scP^1$ and moves to $(q_1',\sigma,q_2)$.
In a state $(q_1,\sigma, q_2) \in Q_1 \times \Sigma \times Q_2$, the \textsf{verifier} selects a $\sigma$ successor $q_2' \in \delta(q_2,\sigma)$ of $q_2$ for $\scP^2$ and moves to $(q_1,q_2')$.

\textsf{Verifier} and \textsf{spoiler} will together produce a play
$(q_0^1,q_0^2)(q_1^1,a_0,q_0^2)
(q_1^1,q_1^2)(q_2^1,a_1,q_1^2)
(q_2^1,q_2^2)(q_3^1,a_2,q_2^2)$
$(q_3^1,q_3^2)(q_4^1,a_3,q_3^2)\ldots$.
The \textsf{verifier} wins if, and only if, the run $q_0^1q_1^1q_2^1q_3^1\ldots$ of $\scP^1$ is rejecting \emph{or} the run  $q_0^2q_1^2q_2^2q_3^2\ldots$ of $\scP^2$ is accepting.

Simulation games have been used to validate GFG-ness right from their introduction \cite{DBLP:conf/csl/HenzingerP06}.

\begin{lemma}
\label{lem:positional}
If \textsf{verifier} wins the $\scP^2 \;\mathsf{simulates}\; \scP^1$ game, then she wins positionally, and checking if she wins is in NP.
\end{lemma}

This is a standard inclusion game, and similar games have e.g.\ been used in \cite{DBLP:conf/icalp/KuperbergS15}.

\begin{proof}
The \textsf{verifier} plays a game with two disjunctive (from \textsf{verifier}'s perspective) parity conditions (as the complement of a parity condition is a parity condition).
A parity condition is in particular a Rabin condition, and the disjunction of Rabin conditions is still a Rabin condition.
Thus, if the verifier can meet her parity objective, she can do so positionally%
\footnote{A strategy is called positional if it only depends on the current state, not on the history of how one got there.}
\cite{Emerson/85/tableaux}.
Thus, it suffices to guess the winning strategy of the \textsf{verifier}, and then check (in P)%
\footnote{This problem is actually in NL, as the falsifier can guess the pair of winning priorities and guess a lasso-like path with an initial part,
and a repeating part that starts and ends in the same state and has the correct dominating priorities for both parity conditions.
(This does not have to be a cycle, as it might be necessary to visit a state once for establishing the dominating priority for either parity condition.)}
if the \textsf{falsifier} wins his resulting one player game with two conjunctive parity conditions.
\end{proof}

\begin{lemma}
\label{lem:inclusion}
Given an NPA $\scP^1$ and a good-for-games NPA $\scP^2$, checking $\scL(\scP^1) \subseteq \scL(\scP^2)$ is in NP.
\end{lemma}

\begin{proof}
Consider the `$\scP^2 \;\mathsf{simulates}\; \scP^1$' game played on an NPA $\scP^1 = (\Sigma,Q_1,q_0^1,\delta_1,\pi_1)$
and a good-for-games NPA $\scP^2 = (\Sigma,Q_2,q_0^2,\delta_2,\pi_2)$.

We first show that \textsf{spoiler} wins this if there is a word $\alpha \in \scL(\scP^1)\setminus\scL(\scP^2)$:
in this case, \textsf{spoiler} can guess such a word alongside an accepting run for $\scP^1$ for $\alpha$---note that there is no accepting run of $\scP^2$ for $\alpha$, as $\alpha \notin \scL(\scP^1)$.

We finally show that \textsf{verifier} wins this game if $\scL(\scP^2)\supseteq\scL(\scP^1)$.
In this case,  \textsf{verifier} can construct the run $q_0^2q_1^2q_2^2q_3^2\ldots$ on the word $\alpha$ the \textsf{spoiler} successively produces.
Moreover, as $\scP^2$ is good-for-games, the \textsf{verifier} can do this independent of the transitions the \textsf{spoiler} selects, 
basing her choices instead on her good-for-games strategy $\nu^2$.
If $\alpha$ is in $\scL(\scP^2)$, then $q_0^2q_1^2q_2^2q_3^2\ldots$ is accepting and \textsf{verifier} wins.
If $\alpha$ is not in $\scL(\scP^2)$, then $\alpha$ is not in $\scL(\scP^1)\subseteq \scL(\scP^2)$ either;
thus $q_0^1q_1^1q_2^1q_3^1\ldots$ is rejecting and \textsf{verifier} wins.
\end{proof}

\begin{theorem}
\label{theo:good-for-games}
Given an NPA $\scP^1$ and a good-for-games NPA $\scP^2$,
checking if $\scL(\scP^1)$ is good-for-games \emph{and} satisfies $\scL(\scP^1)=\scL(\scP^2)$ is in NP.
\end{theorem}

A simplar game is used in \cite{DBLP:conf/icalp/KuperbergS15} to establish an EXPTIME upper bound for establishing EXPTIME inclusion.
The new observations here ate the inclusions in NP.

\begin{proof}
We first use the previous lemma to check $\scL(\scP^1) \subseteq \scL(\scP^2)$ in NP.
For the rest of the proof, we assume that this test has been passed, such that $\scL(\scP^1) \subseteq \scL(\scP^2)$ was established.

We then play the same game with inverse roles, i.e.\ the `$\scP^1 \;\mathsf{simulates}\; \scP^2$' game.
The question if \textsf{verifier} wins is again in NP.

If $\scL(\scP^1) \neq \scL(\scP^2)$ holds, then the already established
$\scL(\scP^1) \subseteq \scL(\scP^2)$ entails that there is a word $\alpha \in \scL(\scP^2) \setminus \scL(\scP^1)$.
In this case \textsf{spoiler} can win by guessing such a word $\alpha \in \scL(\scP^2) \setminus \scL(\scP^1)$ alongside an
accepting run for $\scP^2$ for $\alpha$---note that there is no accepting run of $\scP^1$ for $\alpha$ in this case,
regardless of whether or not $\scP^1$ is good-for games.

If $\scP^1$ is good-for-games \emph{and} $\scL(\scP^1)=\scL(\scP^2)$ holds, then \textsf{verifier} wins
(because $\scL(\scP^2) \subseteq \scL(\scP^1)$ can be verified in NP using Lemma \ref{lem:inclusion}).

Finally, if $\scL(\scP^1)=\scL(\scP^2)$ holds and \textsf{verifier} wins, then  $\scP^1$ is good-for-games: this is because a winning strategy---like
the positional strategy that exists (Lemma \ref{lem:positional})---for \textsf{verifier} in the `$\scP^1 \;\mathsf{simulates}\; \scP^2$' game
transforms a good-for-games strategy $\nu^2$ for $\scP^2$ into a good-for-games strategy $\nu^1$ for $\scP^1$,
and $\scP^1$ can simply emulate the behaviour of $\scP^2$ using the (positional) winning strategy from of the \textsf{verifier}.
\end{proof}

This provides all upper bounds of Theorem \ref{theo:main}.

\begin{corollary}
\label{cor:inNP}
All problems from Theorem \ref{theo:main} can be solved in NP.
\end{corollary}

Note that this does not result in a test whether or not a given automaton is good-for-games, it merely allows, given a good-for-games automaton,
to validate that a second NPA is both: good-for-games \emph{and} language equivalent.
% The small remaining challenge is to establish that, when we already have a good-for-games automaton, we can check
% if a counterpart is good-for-games \emph{and} language equivalent in NP, without the need of being able to check the two properties individually.

For \buchi\ \cite{DBLP:conf/fsttcs/BagnolK18} and \cobuchi\ automata \cite{DBLP:conf/icalp/KuperbergS15},
it is tractable to check whether or not an automaton is good-for-games.

\section{NP Hardness}

In this section we generalise the hardness argument for the minimality of deterministic \buchi\ and \cobuchi\ automata from \cite{Schewe/10/minimise}.
It lifts the reduction from the problem of finding a minimal vertex cover of a graph to the
minimisation of deterministic \buchi\ automata to a reduction to the minimisation of good-for-games automata.
The reduction first defines the characteristic language of a simple connected graph;
for technical convenience it assumes a distinguished initial vertex.

The states of a good-for-games \buchi\ automaton that recognises this characteristic language must satisfy side-constraints,
which imply that it has at least $2n+k$ states, where $n$ is the number of vertices of the graph, and $k$ is the size of its minimal vertex cover.
(For a given a vertex cover of size $k$, it is simple to construct a deterministic \buchi\ automaton of size $2n+k$
that recognises the characteristic language of this graph.)
Consequently, minimising the automaton defined by the trivial vertex cover can be used to determine a minimal vertex cover for this graph,
which concluded the reduction for deterministic automata.

We repeat the argument, showing that a minimal deterministic automaton for this language is also a minimal good-for-games automaton (when measured in the number of states).
As the automaton is deterministic and state minimal, it is also minimal in the size of the transition table.

Finally we show how to adjust the argument for minimal \cobuchi\ automata, which (different to deterministic automata,
where one can simply use the dual automaton) requires a small adjustment in the definition of the characteristic language for good-for-games automata.

Returning to the reduction known from deterministic automata, we call a non-trivial ($|V|>1$) simple undirected connected graph $\G_{v_0}=(V,E)$
with a distinguished initial vertex $v_0\in V$ \emph{nice}.
The restriction to nice graphs leaves the problem of finding a minimal vertex cover NP-complete.

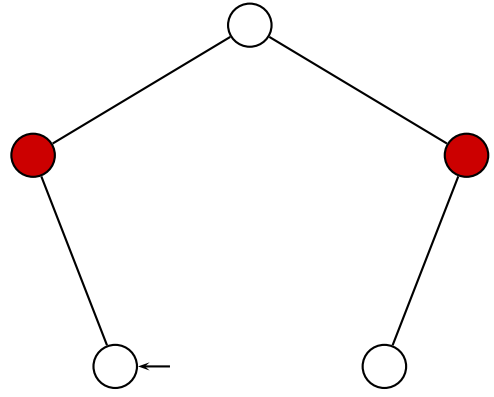
\begin{figure}

\begin{center}
\psset{xunit=30mm,yunit=25mm}
\begin{pspicture}(0,1.2)(0,-1)

\cnode(-0.59,-0.81){3mm}{1}
\cnode[fillstyle=solid,fillcolor=darkred](-0.95,0.31){3mm}{2}
\cnode(0,1){3mm}{3}
\cnode[fillstyle=solid,fillcolor=darkred](0.95,0.31){3mm}{4}
\cnode(0.59,-0.81){3mm}{5}

\pnode(-0.35,-0.81){0}

\ncline{->}{0}{1}
\ncline{-}{1}{2}
\ncline{-}{2}{3}
\ncline{-}{3}{4}
\ncline{-}{4}{5}

\end{pspicture}
\end{center}

\caption{%
A nice graph (a connected graph with a dedicated initial vertex) with a $2$ vertex cover (in red).
\emph{Is a nice graph $k$ coverable?} is a NP-complete problem.}
\label{fig:cover}
\end{figure}

\begin{lemma}
\cite{Schewe/10/minimise}\label{lem:NPC}
The problem of checking whether a nice graph $\G_{v_0}$ has a vertex cover of size $k$ is NP-complete.
\end{lemma}

We define the \emph{characteristic language} $\scL(\G_{v_0})$ of a nice graph $\G_{v_0}$ as the $\omega$-language over $V_\natural=V\uplus\{\natural\}$ (where $\natural$ indicates a stop of the evaluation in the next step---it can be read `stop') consisting of
\begin{enumerate}
\item all $\omega$-words of the form ${v_0}^*{v_1}^+ {v_2}^+ {v_3}^+ {v_4}^+ \ldots \in V^\omega$ with $\{v_{i-1},v_i\}\in E$ for all $i \in \mathbb N$, (words where $v_0,v_1,v_2,\ldots$ form an infinite path in $\G_{v_0}$), and 

\item all $\omega$-words that start with\footnote{this includes words that start with $\natural v_0$} ${v_0}^*{v_1}^+ {v_2}^+ \ldots {v_n}^+ \natural v_n \in {V_\natural}^*$ with $n \in \mathbb N_0$ and $\{v_{i-1},v_i\}\in E$ for all $i \in \mathbb N$.
(Words where $v_0,v_1,v_2,\ldots,v_n$ form a finite---and potentially trivial---path in $\G_{v_0}$, followed by a $\natural$ sign, followed by the last vertex of the path $v_0,v_1,v_2,\ldots,v_n$, and by $v_0$ if $\natural$ was the first letter.)
\end{enumerate}
We call the $\omega$-words in (1) \emph{trace-words}, and those in (2) \emph{$\natural$-words}. The trace-words are in $V^\omega$, while the $\natural$-words are in ${V_\natural}^\omega \setminus V^\omega$.

\begin{figure}[t]

\begin{center}
\psset{xunit=30mm,yunit=25mm}
\begin{pspicture}(0,1.2)(0,-1)

\cnode[fillstyle=solid,fillcolor=darkred](-0.59,-0.81){3mm}{1}
\cnode[fillstyle=solid,fillcolor=darkgreen](-0.95,0.31){3mm}{2}
\cnode[fillstyle=solid,fillcolor=darkblue](0,1){3mm}{3}
\cnode(0.95,0.31){3mm}{4}
\cnode[fillstyle=solid,fillcolor=darkyellow](0.59,-0.81){3mm}{5}
\pnode(-0.35,-0.81){0}

\cnode[fillstyle=solid,fillcolor=darkred](-0.4,-0.54){3mm}{1a}
\cnode[fillstyle=solid,fillcolor=darkgreen](-0.63,0.2){3mm}{2a}
\cnode[fillstyle=solid,fillcolor=darkblue](0,0.67){3mm}{3a}
\cnode(0.63,0.2){3mm}{4a}
\cnode[fillstyle=solid,fillcolor=darkyellow](0.4,-0.54){3mm}{5a}

\ncline{->}{0}{1}
\ncline{<->}{1}{2}
\ncline{<->}{2}{3}
\ncline{<->}{3}{4}
\ncline{<->}{4}{5}

\nccircle[angleA=180]{->}{1}{.235}
\nccircle[angleA=0]{->}{2}{.235}
\nccircle[angleA=-70]{->}{3}{.235}
\nccircle[angleA=0]{->}{4}{.235}
\nccircle[angleA=180]{->}{5}{.235}

\ncline{->}{1}{1a}
\naput{\tiny$\natural$}
\ncline{->}{2}{2a}
\naput{\tiny$\natural$}
\ncline{->}{3}{3a}
\naput{\tiny$\natural$}
\ncline{->}{4}{4a}
\naput{\tiny$\natural$}
\ncline{->}{5}{5a}
\naput{\tiny$\natural$}
\end{pspicture}
\end{center}

\caption{%
An automaton that accepts the $\natural$-words.
The different colours mark the vertices, and the colour of the outer vertices intuitively reflects the \emph{previous} colour/vertex seen, a value which is initialised to the colour of the dedicated initial vertex of the nice graph (in this case, {\color{darkred}\textbullet}).
If the automaton reads a vertex (here identified by its colour), which identifies either the current vertex or a vertex adjacent to it, it updates the stored vertex to the one it has read.
If it reads a different vertex that is not adjacent, it blocks (moves to $\bot$).\newline
When reading $\natural$, it moves to the inner vertex while keeping the stored colour/vertex.
From an inner vertex, it accepts (moves to $\top$) if it sees the stored vertex next, and blocks (moves to $\bot$) otherwise.\newline
A word  {\color{darkred}\textbullet\textbullet\textbullet
 \color{darkgreen}\textbullet\textbullet
 \color{darkred}\textbullet
 \color{darkgreen}\textbullet
 \color{darkblue}\textbullet
 \color{black}$\circ\circ$%
 \color{darkblue}\textbullet\textbullet
 {\color{black}$\natural$}\textbullet } \ldots, for example, is accepted, while the words
 {\color{darkred}\textbullet\textbullet\textbullet
 \color{darkgreen}\textbullet\textbullet
 \color{darkred}\textbullet
 \color{darkgreen}\textbullet
 \color{darkblue}\textbullet
 \color{black}$\circ\circ$%
 \color{darkblue}\textbullet\textbullet
 \color{black}$\natural$%
 \color{darkyellow}\textbullet } \ldots
 (wrong colour after $\natural$) and
 {\color{darkred}\textbullet\textbullet\textbullet
 \color{darkgreen}\textbullet\textbullet
 \color{darkred}\textbullet
 \color{darkgreen}\textbullet
 \color{darkblue}\textbullet
 \color{darkyellow}\textbullet
 \color{black}$\circ$%
 \color{darkblue}\textbullet\textbullet
 {\color{black}$\natural$}\textbullet } \ldots
 ({\color{darkblue}\textbullet} and
 {\color{darkyellow}\textbullet} are not adjacent) are rejected.
 }
\label{fig:natural}
\end{figure}

Let $\B$ be a parity good-for-games automaton that recognises the characteristic language of $\G_{v_0}=(V,E)$.
We call a state of $\B$
\begin{itemize}
\item a \emph{$v$-state} if it can be reached upon an input word ${v_0}^*{v_1}^+ {v_2}^+ \ldots {v_n}^+ \in {V}^*$, with $n \in \mathbb N_0$ and $\{v_{i-1},v_i\}\in E$ for all $i \in \mathbb N$, that ends in $v=v_n$ (in particular, the initial state of $\B$ is a $v_0$-state), and
\item a \emph{$v\natural$-state} if it can be reached from a $v$-state upon reading a $\natural$ sign.
\end{itemize}
We call the union over all $v$-states the set of \emph{vertex-states}, and the union over all $v\natural$-states the set of $\natural$-states.

\begin{lemma}
\label{lem:large}
Let $\mathcal G_{v_0}=(V,E)$ be a nice graph with initial vertex $v_0$, and let $\B=(V,Q,q_0,\delta,\pi)$ be a good-for-games parity automaton that recognises the characteristic language of $\G_{v_0}$.
Then the following holds.
\begin{enumerate}
 \item for all $v$ in $V$, there is a $v$-state from which all words that start with $\natural v$ are accepted---we call these states the \emph{core $v$-states};
 \item for all $v$ in $V$, there is a core $v$-state with an odd priority;
 \item for all $v \in V$ and $w \in V_\natural$ with $v\neq w$ and for every $v$-state $q_v$,
 words that start with $\natural w$ are not in the language of $\B_{q_v}$;

 \item for all $v$ in $V$, there is a $\natural v$-state from which all words that start with $v$ are accepted---we call these states the \emph{core $\natural v$-states};
 \item for all $v$ in $V$ and $w \in V_\natural$ with $v\neq w$ and for every $v$-state $q_{\natural v}$,
 words that start with $w$ are not in the language of $\B_{q_{\natural v}}$; and

 \item for every edge  $\{v,w\}\in E$, there is a $v$-state or a $w$-state with an even priority.

 \end{enumerate}
\end{lemma}
\begin{figure}[t]

\begin{center}
\psset{xunit=30mm,yunit=25mm}
\begin{pspicture}(0,1.2)(0,-1)

\cnode[fillstyle=solid,fillcolor=darkred](-0.59,-0.81){3mm}{1}
\cnode[fillstyle=solid,fillcolor=darkgreen](-0.95,0.31){3mm}{2}
\cnode[fillstyle=solid,fillcolor=darkblue](0,1){3mm}{3}
\cnode(0.95,0.31){3mm}{4}
\cnode[fillstyle=solid,fillcolor=darkyellow](0.59,-0.81){3mm}{5}
\pnode(-0.35,-0.81){0}

\cnode[fillstyle=solid,fillcolor=darkred](-0.4,-0.54){3mm}{1a}
\cnode[fillstyle=solid,fillcolor=darkgreen](-0.63,0.2){3mm}{2a}
\cnode[fillstyle=solid,fillcolor=darkblue](0,0.67){3mm}{3a}
\cnode(0.63,0.2){3mm}{4a}
\cnode[fillstyle=solid,fillcolor=darkyellow](0.4,-0.54){3mm}{5a}

\cnode[fillstyle=solid,fillcolor=darkgreen,doubleline=true](-1.33,0.434){3mm}{2b}
\cnode[doubleline=true](1.33,0.434){3mm}{4b}

\ncline{->}{0}{1}
\ncline{<-}{1}{2}
\ncline{->}{2}{3}
\ncline{<-}{3}{4}
\ncline{->}{4}{5}

\ncarc[arcangle=40]{<->}{1}{2b}
\ncarc[arcangle=40]{<->}{2b}{3}
\ncarc[arcangle=40]{<->}{3}{4b}
\ncarc[arcangle=40]{<->}{4b}{5}

\ncline{->}{2b}{2}
\ncline{->}{4b}{4}

\nccircle[angleA=180]{->}{1}{.235}
\nccircle[angleA=0]{->}{2}{.235}
\nccircle[angleA=-70]{->}{3}{.235}
\nccircle[angleA=0]{->}{4}{.235}
\nccircle[angleA=180]{->}{5}{.235}

\ncline{->}{1}{1a}
\naput{\tiny$\natural$}
\ncline{->}{2}{2a}
\naput{\tiny$\natural$}
\ncline{->}{3}{3a}
\naput{\tiny$\natural$}
\ncline{->}{4}{4a}
\naput{\tiny$\natural$}
\ncline{->}{5}{5a}
\naput{\tiny$\natural$}

\ncarc[arcangle=-40]{->}{2b}{2a}
\nbput{\tiny$\natural$}
\ncarc[arcangle=40]{->}{4b}{4a}
\naput{\tiny$\natural$}
\end{pspicture}
\end{center}

\caption{%
An automaton that also accepts the trace-words.
For a nice graph  $\mathcal G_{v_0}=(V,E)$, it needs $|V|$ states reached after reading (the first) $\natural$,
$|V|$ non-final states reachable prior to reading the first $\natural$, and, broadly speaking, sufficiently many final states, such that they form a cover.
The automaton shown here is defined by the cover shown in Figure \ref{fig:cover}.}
\label{fig:trace}
\end{figure}
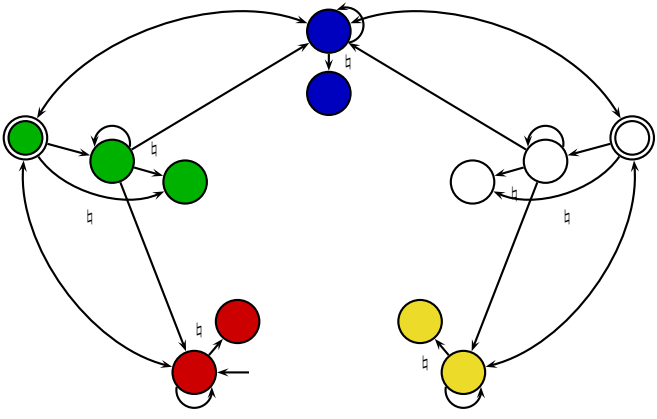
\begin{proof}
\begin{enumerate}
\item Let $v=v_n$ and let $v_0,v_1,v_2,...,v_n$ be a path in $\G_{v_0}$.
As $\B$ recognises $\scL( \G_{v_0})$ and is good-for-games, it must, after having read the first $n+1$ or more letters of an input word $v_0,v_1,v_2,...,{v_n}^\omega$ (using its good-for-games strategy $\nu$),
with $\{v_i,v_{i+1}\} \in E$ for all $i<n$, be in a core $v$-state, as words that start with this and continue with $\natural v$ are in $\scL( \G_{v_0})$.

\item Furthermore, the run $\B$ produces (using $\nu$) for $v_0,v_1,v_2,...,{v_n}^\omega$ has a dominating priority determined by its tail of
core $v$-states, and the core $v$-state with the highest priority that occurs infinitely many times must have an odd priority
(as the word is not $\scL( \G_{v_0})$).
Consequently, there must be at least one core $v$-state with an odd priority.

\item If (3) does not hold, a witness would provide a word accepted by $\B$ but not in  $\scL( \G_{v_0})$.

\item Let $v=v_n$ and let $v_0,v_1,v_2,...,v_n$ be a path in $\G_{v_0}$.
As $\B$ recognises $\scL( \G_{v_0})$ and is good-for-games, it must, after having read the first $n+2$ letters of an input word that starts with $v_0,v_1,v_2,...,v_n,\natural,^\omega$ (using its good-for-games strategy $\nu$),  with $\{v_i,v_{i+1}\} \in E$ for all $i<n$, be in a core $\natural v$-state, as words that start with this and continue with $v$ are in $\scL( \G_{v_0})$.

\item If (5) does not hold, a witness would provide a word accepted by $\B$ but not in  $\scL( \G_{v_0})$.

\item Let us consider an arbitrary edge $\{v,w\} \in E$, $v=v_n$, and the run of $\B$ (following $\nu$) on
$v_0,v_1,v_2,\ldots, v_n, (w,v)^\omega$ in $\scL(\G_{v_0})$ (i.e.\ for all $i<n.\ \{v_i,v_{i+1}\}\in E$).

The run must be accepting, and, as argued in (1), once the word alternates between $v$ and $w$, the run alternates between core $v$-states and core $w$-states.
Thus, the core $v$-state \emph{or} the core $w$-state with the highest priority that occurs infinitely often must have an even priority.
\end{enumerate}
% This together implies that there must be at least $|V|$ core $\natural$ states, $|V|$ rejecting core vertex states, and accepting core vertex states that define a cover; and they must be pairwise disjoint.
\end{proof}

The sixth claim implies that the set $C$ of vertices with a core vertex-state with even priority is a vertex cover of $\mathcal G_{v_0}=(V,E)$.
Thus, $\B$ has at least $|C|$ core vertex states with an even priority.
(1--3) provide that $\B$ has at least $|V|$ vertex-states with odd priority, and it follows
with (4+5) that there are $|V|$ core $\natural$-states that are disjoint from the core vertex-states:

\begin{corollary}
\label{cor:atLeast}
For a good-for-games parity automaton $\B=(V,Q,q_0,\delta,\pi)$ that recognises the characteristic language of a nice graph $\mathcal G_{v_0}=(V,E)$ with initial vertex $v_0$,
the set $C=\{v\in V \mid$ there is a $v$-state with an even priority$\}$ is a vertex cover of $\mathcal G_{v_0}$, and $\B$ has at least $2|V| + |C|$ states.
\qed
\end{corollary}

It is not hard to define, for a given nice graph $\mathcal G_{v_0}=(V,E)$ with vertex cover $C$, a \emph{deterministic} \buchi\ automaton
$\B^{\G_{v_0}}_C= (V_\natural,(V\times\{n,\natural\}) \uplus (C \times \{f\}),(v_0,n),\bar{\delta},(C\times\{f\})\uplus \{\top\})$
with $2|V|+|C|$ states that recognises the characteristic language of $\mathcal G_{v_0}$ \cite{Schewe/10/minimise}.
(The $n$ and $f$ in the state refer to non-final and final, respectively.)
We simply choose
\begin{itemize}
\item $\bar{\delta}\big((v,n),v'\big) = (v',f)$ if $\{v,v'\} \in E$ and $v'\in C$,

$\bar{\delta}\big((v,n),v'\big) = (v',n)$ if $\{v,v'\} \in E$ and $v'\notin C$,

$\bar{\delta}\big((v,n),v'\big) = (v,n)$ if $v=v'$,

$\bar{\delta}\big((v,n),v'\big) = (v,\natural)$ if $v'=\natural$, and

$\bar{\delta}\big((v,n),v'\big)=\bot$ otherwise;

\item $\bar{\delta}\big((v,f),v'\big) = \bar{\delta}\big((v,n),v'\big)$, and
\item $\bar{\delta}\big((v,\natural),v\big) = \top$ and $\bar{\delta}\big((v,\natural),v'\big) = \bot$ for $v' \neq v$.
\end{itemize}

$\B^{\G_{v_0}}_C$ simply has one $v\natural$-state for each vertex $v\in V$ of $\mathcal G_{v_0}$, one accepting $v$-state for each vertex in the vertex cover $C$, and one rejecting $v$-vertex for each vertex $v\in V$ of $\mathcal G_{v_0}$.
It moves to the accepting copy of a vertex state $v$ only upon taking an edge to $v$, but not on a repetition of~$v$.

\begin{lemma}
\cite{Schewe/10/minimise}
\label{lem:correct}
For a nice graph $\G_{v_0}=(V,E)$ with initial vertex $v_0$ and vertex cover $C$, the \buchi\ automaton $\B^{\G_{v_0}}_C$ recognises the characteristic language of $\mathcal G_{v_0}$.
\end{lemma}

Corollary \ref{cor:atLeast} and Lemma \ref{lem:correct} immediately imply:

\begin{corollary}
\label{cor:reduce}
Let $C$ be a minimal vertex cover of a nice graph $\mathcal G_{v_0}=(V,E)$. Then $\B^{\G_{v_0}}_C$ is a minimal deterministic \buchi\ automaton
that recognises the characteristic language of $\mathcal G_{v_0}$,
and there is no good-for-games parity automaton with less states than  $\B^{\G_{v_0}}_C$ that recognises the same language.
\qed
\end{corollary}

We change the characteristic language to the \emph{adjusted language}  $\scL'(\G_{v_0})$ of a nice graph $\G_{v_0}$
as the $\omega$-language over $V_\natural=V\uplus\{\natural\}$ that consists of
\begin{enumerate}
\item all $\omega$-words of the form ${v_0}^*{v_1}^+ {v_2}^+ {v_3}^+ {v_4}^+ \ldots {v_n}^\omega \in V^\omega$ with $\{v_i,v_{i+1}\}\in E$ for all $i < n$, (words where $v_0,v_1,v_2,\ldots,v_n$ form a finite (possibly trivial) path in $\G_{v_0}$, and 

\item all $\omega$-words that start with\footnote{this includes words that start with $\natural v_0$} ${v_0}^*{v_1}^+ {v_2}^+ \ldots {v_n}^+ \natural v_n \in {V_\natural}^*$ with $n \in \mathbb N_0$ and $\{v_{i-1},v_i\}\in E$ for all $i \in \mathbb N$.
(Words where $v_0,v_1,v_2,\ldots,v_n$ form a finite---and potentially trivial---path in $\G_{v_0}$, followed by a $\natural$ sign, followed by the last vertex of the path $v_0,v_1,v_2,\ldots,v_n$, and by $v_0$ if $\natural$ was the first letter.)
\end{enumerate}

\begin{lemma}
\label{lem:colarge}
Let $\mathcal G_{v_0}=(V,E)$ be a nice graph with initial vertex $v_0$, and let $\B=(V,Q,q_0,\delta,\pi)$ be a good-for-games parity automaton that recognises the adjusted language $\scL'(\G_{v_0})$ of  $\G_{v_0}$.
Then the following holds.
\begin{enumerate}
 \item for all $v$ in $V$, there is a $v$-state from which all words that start with $\natural v$ are accepted---we call these states the \emph{core $v$-states};
 \item for all $v$ in $V$, there is a core $v$-state with an \textbf{even} priority;
 \item for all $v \in V$ and $w \in V_\natural$ with $v\neq w$ and for every $v$-state $q_v$,
 words that start with $\natural w$ are not in the language of $\B_{q_v}$;

 \item for all $v$ in $V$, there is a $\natural v$-state from which all words that start with $v$ are accepted---we call these states the \emph{core $\natural v$-states};
 \item for all $v$ in $V$ and $w \in V_\natural$ with $v\neq w$ and for every $v$-state $q_{\natural v}$,
 words that start with $w$ are not in the language of $\B_{q_{\natural v}}$; and

 \item for every edge  $\{v,w\}\in E$, there is a $v$-state or a $w$-state with an \textbf{odd} priority.

 \end{enumerate}
\end{lemma}

The changes in the proof compared to Lemma \ref{lem:large} are simply to replace even and odd accordingly.

With the same argument as before we get the same corollary:

\begin{corollary}
\label{cor:coatLeast}
For a good-for-games parity automaton that recognises the adjusted language of a nice graph $\mathcal G_{v_0}=(V,E)$ with initial vertex $v_0$,
the set $C=\{v\in V \mid$ there is a $v$-state with an even priority$\}$ is a vertex cover of $\mathcal G_{v_0}$,
and $\B$ has at least $2|V| + |C|$ states.
\end{corollary}

\begin{lemma}
\label{lem:cocorrect}
For a nice graph $\G_{v_0}=(V,E)$ with initial vertex $v_0$ and vertex cover $C$, the
\cobuchi\ automaton%
\footnote{The automaton is the same as before, but read as a \cobuchi\ automaton.}
$\B^{\G_{v_0}}_C$ recognises the adjusted language of $\mathcal G_{v_0}$.
\end{lemma}

\begin{proof}
We argue separately for trace-words and $\natural$-words accepted by $\B^{\G_{v_0}}_C$ are exactly the trace-words and $\natural$-words in $\scL'(\mathcal G_{v_0})$

For a \emph{trace-word} $\alpha=v_1v_2v_3\ldots \in V^\omega$, $\B^{\G_{v_0}}_C$ has the run $(v_0,n)(v_1,x_1)(v_2,x_2)(v_3,x_3)\ldots$
(with $x_i \in \{n,f\}$ for all $i \in \mathbb N$) if, for all $i\in \mathbb N$, either 
$v_{i-1}=v_i$ or $\{v_{i-1},v_i\}\in E$ holds; otherwise the automaton blocks (has a tail of $\bot$ states) at $i_{\min}$-th letter, where $i_{\min}$ is the
minimal $i$ such that $v_{i-1}\neq v_i$ and $\{v_{i-1},v_i\}\notin E$.
A trace-word where the automaton blocks is rejected by $\B^{\G_{v_0}}_C$ and not in $\scL'(\G_{v_0})$.

We now consider those trace-words, for which $\B^{\G_{v_0}}_C$ does not block.
For these words, we call the set $I = \{ i \in \mathbb N \mid \{v_{i-1},v_i\} \in E\big\}$ transition indices.
Now $\alpha \in \scL'(\G_{v_0})$ holds if, and only if, $I$ is finite.
If $I$ is finite, we call its maximal element $i_{\max}$, and set $i_{\max}$ to $0$ if $I$ is empty.
The run of $\B^{\G_{v_0}}_C$ on $\alpha$ is then
$(v_0,n)(v_1,x_1)\ldots(v_{i_{\max}-1},x_{i_{\max}-1})(v_{i_{\max}},x_{i_{\max}})(v_{i_{\max}},n)^\omega$;
it has a tail of non-final states $(v_{i_{\max}},n)$, and $\alpha$ is therefore accepted by $\B^{\G_{v_0}}_C$.

If $I$ is infinite, we use the infinite ascending chain $i_1 < i_2 < i_3 < \ldots$ with $I = \{i_n \mid n \in \mathbb N\}$.
Then, for all $k\in \mathbb N$, $v_{i_k-1} \neq v_{i_k} = v_{i_{k+1}-1} \neq v_{i_{k+1}}$ holds and $\{v_{i_k},v_{i_{k+1}}\} \in E$.
$\{v_{i_k},v_{i_{k+1}}\} \in E$ entails that the cover $C$ must contain $v_{i_k}$ or $v_{i_{k+1}}$,
and it follows with $v_{i_k-1} \neq v_{i_k}$ \emph{and} $v_{i_{k+1}-1} \neq v_{i_{k+1}}$ that the respective position in the run is
$(v_{i_k},f)$ \emph{or} $(v_{i_{k+1}},f)$ (in other words: $x_{i_k} = f$ or $x_{i_{k+1}}=f$).
Thus, the run contains infinitely many final states and is therefore rejecting.

Thus, we have shown that $\B^{\G_{v_0}}_C$ accepts the right set of trace-words.
We now continue with the simpler proof that it accepts the right set of $\natural$-words.

First, words starting with $\natural v_0$ are accepted and in $ \scL'(\G_{v_0})$, while
words starting with $\natural v$ and $v \neq v_0$ are rejected and not in $\scL'(\G_{v_0})$.

A \emph{$\natural$-word} that starts with $\alpha=v_1v_2v_3\ldots v_n \natural w \in V^+ \natural V_\natural$
is in  $\scL'(\G_{v_0})$ if, and only if,
\begin{enumerate}
 \item $v_{i-1}=v_i$ or $\{v_{i-1},v_i\}\in E$ holds for all $i \leq n$, and
 \item $v_n = w$.
\end{enumerate}
If they both hold, the (accepting) run of $\B^{\G_{v_0}}_C$ has the form
$(v_0,n)(v_1,x_1)(v_2,x_2)(v_3,x_3)\ldots(v_n,x_n)(v_n,\natural)\top^\omega$.

If (1) holds but (2) does not, the (rejecting) run of $\B^{\G_{v_0}}_C$ has the form
$(v_0,n)(v_1,x_1)(v_2,x_2)(v_3,x_3)\ldots(v_n,x_n)(v_n,\natural)\bot^\omega$.

If (1) does not hold and $k \leq n$ is the smallest index with $v_{i-1}\neq v_i$ and $\{v_{i-1},v_i\}\notin E$,
the (rejecting) run of $\B^{\G_{v_0}}_C$ has the form
$(v_0,n)(v_1,x_1)(v_2,x_2)(v_3,x_3)\ldots(v_{k-1},x_{k-1})\bot^\omega$.

As this covers all cases, we get $\scL(\B^{\G_{v_0}}_C) = \scL'(\mathcal G_{v_0})$.
\qed
\end{proof}

Corollary \ref{cor:coatLeast} and Lemma \ref{lem:cocorrect} immediately imply:

\begin{corollary}
\label{cor:coreduce}
Let $C$ be a minimal vertex cover of a nice graph $\mathcal G_{v_0}=(V,E)$. Then $\B^{\G_{v_0}}_C$ is a minimal deterministic \cobuchi\ automaton
that recognises the characteristic language of $\mathcal G_{v_0}$,
and there is no good-for-games parity automaton with less states than  $\B^{\G_{v_0}}_C$ that recognises the same language.
\qed
\end{corollary}

The Corollaries \ref{cor:reduce} and \ref{cor:coreduce} provide us with the hardness result.

\begin{theorem}
\label{theo:hard}
For a
\begin{itemize}
 \item good-for-games \buchi\ automaton
 \item deterministic \buchi\ automaton
\end{itemize}
and a bound $k$, it is NP hard to check if there is a
\begin{itemize}
 \item good-for-games \buchi\ automaton
 \item good-for-games parity automaton
\end{itemize}
with at most
\begin{itemize}
 \item $k$ states
 \item $k$ transitions.
\end{itemize}

For a
\begin{itemize}
 \item good-for-games \cobuchi\ automaton
 \item deterministic \cobuchi\ automaton
\end{itemize}
and a bound $k$, it is NP hard to check if there is a
\begin{itemize}
 \item good-for-games \cobuchi\ automaton
 \item good-for-games parity automaton
\end{itemize}
with at most
\begin{itemize}
 \item $k$ states
 \item $k$ transitions.
\end{itemize}

For a
\begin{itemize}
 \item good-for-games parity automaton
 \item deterministic parity automaton
\end{itemize}
and a bound $k$, it is NP hard to check if there is a
\begin{itemize}
 \item good-for-games parity automaton
\end{itemize}
with at most
\begin{itemize}
 \item $k$ states
 \item $k$ transitions.
\end{itemize}

(These claims hold in all combinations.)
\end{theorem}

\section{Discussion}
We have established that determining if a good-for-games automaton with \buchi, \cobuchi\, or parity condition and
state based acceptance is minimal, or that there is a GFG automaton with size up to $k$, is NP-complete.
Moreover, this holds regardless of whether the starting automaton is given as a (\buchi, \cobuchi, or parity) qood-for-games automaton, or if it presented as
a (\buchi, \cobuchi, or parity) deterministic automaton.

This drags three open questions into the limelight.
The first is the complexity of testing whether or not a given nondeterministic automaton is good-for-games.
Our results give no answer to this question: it simply accepts that a given automaton is good-for-games, and
only guarantees a correct answer if the input is valid.
GFG-ness is, however, known to be tractable for \buchi~\cite{DBLP:conf/fsttcs/BagnolK18} and \cobuchi~\cite{DBLP:conf/icalp/KuperbergS15} automata, and the extension to the more expressive class
parity good-for-games automata is active research.

It also raises the question if the difference is in good-for-games automata being inherently simpler to minimise, or if it is a property of
choosing the less common transition based acceptance:
the second and open challenge is whether the tractability of minimising \cobuchi\ good-for-games automata forebears the
tractability of minimising the general class of parity good-for-games automata, while the third challenge is the question of whether NP hardness
extends to transition based deterministic \buchi, \cobuchi, and parity automata.

The latter two challenges ask for the level of superiority of transition based acceptance: they have the natural advantage of potentially---slightly---higher succinctness,
but a combination with an improved complexity could turn this class into the standard.

In addition to the `transition vs.\ state based acceptance' question, another question is whether or not nondeterminism is the right starting point for GFG-ness,
or if alternation is the better choice \cite{boker_et_al:LIPIcs:2019:10921}. For such alternating automata, most of the succinctness and complexity questions for membership and minimisation are wide open.

\paragraph{\bf Acknowledgments.} Many thanks to Patrick Totzke and Karoliina Lehtinen for valuable feedback, enduring and helping to erase the
errors of draft versions, and pointers to beautiful related works.

% \bibliographystyle{alpha}
% \bibliography{bib}
% 
% \end{document}

\end{document}